\pdfoutput=1 
\documentclass[aps,prl,twocolumn,showpacs,superscriptaddress,letterpaper]{revtex4}  
\usepackage{graphicx}  
\usepackage{dcolumn}   
\usepackage{bm}        
\usepackage{amsmath,latexsym,amssymb,amsthm,revsymb,array}
\usepackage{wrapfig,color}
\usepackage{hyperref}

\newtheorem{theorem}{Theorem}

\newtheorem{proposition}[theorem]{Proposition}

\newcommand{\nc}{\newcommand}
\nc{\SR}{\mathrm{SR}}
\nc{\SE}{\mathrm{SE}}
\nc{\NS}{\mathrm{NS}}

\nc{\pw}{\mathrm{PW}}
\nc{\arbclass}{\mathrm{\Omega}}
\nc{\rnc}{\renewcommand}
\nc{\beq}{\begin{equation}}
\nc{\mc}{\mathcal}
\nc{\eeq}{{\end{equation}}}
\nc{\beqa}{\begin{eqnarray}}
\nc{\eeqa}{\end{eqnarray}}
\nc{\lbar}[1]{\overline{#1}}
\nc{\bra}[1]{\langle#1|}
\nc{\ket}[1]{|#1\rangle}
\nc{\ketbra}[2]{|#1\rangle\!\langle#2|}
\nc{\braket}[2]{\langle#1|#2\rangle}
\nc{\proj}[1]{| #1\rangle\!\langle #1 |}
\nc{\avg}[1]{\langle#1\rangle}
\nc{\Rank}{\operatorname{rank}\,}
\nc{\smfrac}[2]{\mbox{$\frac{#1}{#2}$}}
\nc{\tr}{\operatorname{Tr}}
\nc{\ox}{\otimes}
\nc{\catchset}{T}
\nc{\dg}{\dagger}
\nc{\dn}{\downarrow}
\nc{\cA}{{\cal A}}
\nc{\cB}{{\cal B}}
\nc{\cC}{{\cal C}}
\nc{\cD}{{\cal D}}
\nc{\cE}{{\cal E}}
\nc{\cF}{{\cal F}}
\nc{\cG}{{\cal G}}
\nc{\cH}{{\cal H}}
\nc{\cI}{{\cal I}}
\nc{\cJ}{{\cal J}}
\nc{\cK}{{\cal K}}
\nc{\cL}{{\cal L}}
\nc{\cM}{{\cal M}}
\nc{\cN}{{\cal N}}
\nc{\cO}{{\cal O}}
\nc{\cP}{{\cal P}}
\nc{\cR}{{\cal R}}
\nc{\cS}{{\cal S}}
\nc{\cT}{{\cal T}}
\nc{\cX}{{\cal X}}
\nc{\cY}{{\cal Y}}
\nc{\cZ}{{\cal Z}}
\nc{\csupp}{{\operatorname{csupp}}}
\nc{\qsupp}{{\operatorname{qsupp}}}
\nc{\var}{{\operatorname{var}}}
\nc{\rar}{\rightarrow}
\nc{\lrar}{\longrightarrow}
\nc{\polylog}{{\operatorname{polylog}}}
\nc{\1}{{\mathbbm{1}}}
\nc{\wt}{{\operatorname{wt}}}

\nc{\RR}{{{\mathbb R}}}
\nc{\CC}{{{\mathbb C}}}
\nc{\FF}{{{\mathbb F}}}
\nc{\NN}{{{\mathbb N}}}
\nc{\ZZ}{{{\mathbb Z}}}
\nc{\PP}{{{\mathbb P}}}
\nc{\QQ}{{{\mathbb Q}}}
\nc{\UU}{{{\mathbb U}}}
\nc{\EE}{{{\mathbb E}}}
\nc{\id}{{\operatorname{id}}}

\nc{\CHSH}{{\operatorname{CHSH}}}

\nc{\be}{\begin{equation}}
\nc{\ee}{{\end{equation}}}
\nc{\bea}{\begin{eqnarray}}
\nc{\eea}{\end{eqnarray}}
\nc{\<}{\langle}
\rnc{\>}{\rangle}
\nc{\Hom}[2]{\mbox{Hom}(\CC^{#1},\CC^{#2})}
\nc{\rU}{\mbox{U}}

\nc{\ob}[1]{#1}

\def\SE{\mathrm{SE}}

\begin{document}
	\title{Improving zero-error classical communication with entanglement}
	\author{Toby S. Cubitt}
	\affiliation{Department of Mathematics, University of Bristol, Bristol, BS8 1TW, U.K.}

	\author{Debbie Leung}
	\affiliation{Institute for Quantum Computing, University of Waterloo, Waterloo, N2L 3G1, ON, Canada}

	\author{William Matthews}
	\email{will@northala.net}
	\affiliation{Institute for Quantum Computing, University of Waterloo, Waterloo, N2L 3G1, ON, Canada}

	\author{Andreas Winter}
	\affiliation{Department of Mathematics, University of Bristol, Bristol BS8 1TW, U.K.}
	\affiliation{Centre for Quantum Technologies, National University of Singapore,
	 2 Science Drive 3, Singapore 117542}

	\date{19 February 2010}
	\pacs{03.67.Ac, 03.67.Bg, 89.70.Kn}

	\begin{abstract}
		Given one or more uses of a classical channel, only a certain number of messages can be transmitted with zero probability of error. The study of this number and its asymptotic behaviour constitutes the field of classical zero-error information theory \cite{shannon,ZEIT}, the quantum generalisation of which has started to develop recently \cite*{Med,Duan,CCH-zero,superduper}. We show that, given a single use of certain classical channels, entangled states of a system shared by the sender and receiver can be used to increase the number of (classical) messages which can be sent with no chance of error. In particular, we show how to construct such a channel based on any proof of the Kochen-Specker theorem \cite{KS}. This is a new example of the use of quantum effects to improve the performance of a classical task. We investigate the connection between this phenomenon and that of ``pseudo-telepathy'' games. The use of generalised non-signalling correlations to assist in this task is also considered. In this case, a particularly elegant theory results and, remarkably, it is sometimes possible to transmit information with zero-error using a channel with \emph{no} unassisted zero-error capacity.
	\end{abstract}

	\maketitle

	It is well known that if two parties share an entangled quantum state, they may be able to achieve tasks which would be otherwise impossible. For instance, without communicating they can violate Bell inequalities \cite{Bell}, and with classical communication they can teleport the state of a quantum system \cite{Tele}. Here we show that quantum effects can sometimes give an advantage in the context of zero-error coding \cite{shannon,ZEIT}: A classical channel $\mathcal{N}$ connects a sender (Alice) to a receiver (Bob). It has a finite number of inputs and outputs and its behaviour is fully described by the conditional probability distribution over outputs given the input, i.e.\ it is \emph{discrete} and \emph{memoryless}. Given one use of $\mc{N}$, the maximum number of different messages Alice can send to Bob if there is to be no chance of an error is known as the \emph{one-shot zero-error capacity} of $\mc{N}$.

	The main contribution of this paper is to show that for certain classical channels, entanglement between Alice and Bob can be used to increase the one-shot zero-error capacity for classical messages. This is in contrast to interesting recent work considering zero-error coding for classical and quantum data over \emph{quantum} channels \cite*{Med,Duan,CCH-zero,superduper}. Recall that the use of entanglement \cite{QRST} (and even non-signalling correlations \cite{long}) \emph{cannot} increase the transmission rate if we only demand that the error rate goes to zero in the large block length limit: it remains equal to the normal Shannon capacity \cite{shannon-BIG}.
		
	We briefly review classical zero-error coding, then we show how to construct classical channels where entanglement can increase the one-shot zero-error capacity. We then discuss the relationship of entanglement assisted zero-error coding to ``pseudo-telepathy'' games. After that, we upper bound this entanglement assistance by considering generalised non-signalling correlations, giving a simple formula for the non-signalling assisted zero-error capacity of any channel. This turns out to have an interesting relationship to classical results of Shannon from his original paper \cite{shannon} on zero-error capacities.

	Two input symbols of a channel are \emph{confusable} if the corresponding distributions on output symbols overlap. Shannon introduced the \emph{confusability graph} $G(\mc{N})$ of a classical channel $\mc{N}$: Its vertices are the set of input symbols and they are joined if and only if they are confusable. Classically, a zero-error code is a set of non-confusable inputs. The one-shot zero-error capacity $c_0(\mc{N})$ of a channel $\mc{N}$ is simply the maximum size of such a set. In the language of graph theory, a maximum non--confusable set of inputs is a \emph{maximum independent set} of the confusability graph, and when Bob receives a channel output, the possible inputs are a \emph{clique} in the confusability graph. A channel has \emph{no} unassisted zero-error capacity if and only if its confusability graph is \emph{complete} i.e.\ all vertices are connected.

	It is also useful to define the hypergraph of a channel: A hypergraph is just a set $S$ (the vertices) and a set of subsets of $S$ called the \emph{hyperedges}. The hypergraph of a channel $\mc{N}$ has the set of inputs as vertices and one hyperedge for each of the outputs, which contains all the inputs that have a non-zero probability of causing that output; we denote it $H(\mc{N})$. See Figure \ref{fig:graphs} for an example illustrating the confusability graph and channel hypergraph.
	
	\begin{figure}
	  \includegraphics[scale=0.9]{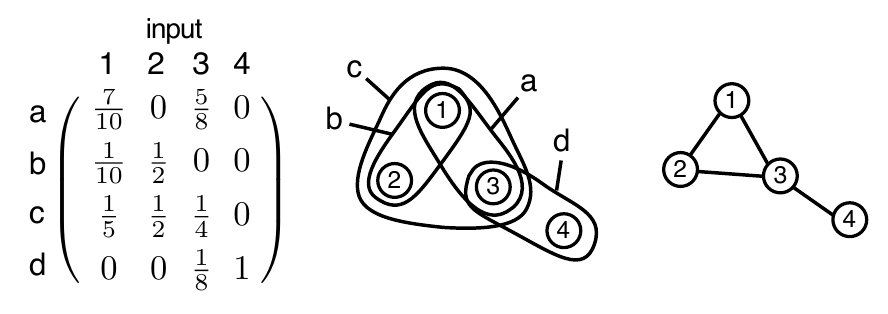}
	  \caption{From left to right: The conditional probability matrix of a classical channel $\mc{N}$ with inputs in $\{1,2,3,4\}$ and outputs in $\{a,b,c,d\}$; Its hypergraph $H(\mc{N})$, with the hyperedges labelled by the corresponding outputs; Its confusability graph $G(\mc{N})$. From $G(\mc{N})$ it is easy to see that inputs 1 and 4 form a maximum non--confusable set (as do 2 and 4) so $c_0(\mc{N}) = 2$.}
	  \label{fig:graphs}
	\end{figure}

In this work we deal with \emph{correlations} (bipartite conditional probability distributions) in the classes $\SR$, $\SE$ and $\NS$: Correlations belong to $\rm{SR}$ if and only if they can be obtained using (classical) Shared Randomness (and local operations); to $\rm{SE}$ (Shared Entanglement) if and only if they can be realised by local operations on a shared quantum state; and to $\rm{NS}$ if and only if the correlation is Non--Signalling (meaning that the marginal distribution on the output of each party is independent of the other party's input). Each class in this list strictly contains the previous one.
	We denote the maximum number of messages which can be sent without error by a single use of $\mathcal{N}$ when any correlation in class $\arbclass$ can be used by $c_{\arbclass}(\mathcal{N})$. The corresponding limiting rate to send zero-error bits is $C_{\arbclass}(\mathcal{N}) := \lim_{n \to \infty} \frac{1}{n} \log c_{\arbclass}(\mathcal{N}^{\ox n})$. A simple convexity argument shows that shared randomness between sender and receiver cannot help, so $c_{\SR}(\mc{N}) = c_0(\mc{N})$ for all channels. In constrast, we will next show how to construct channels $\mc{N}$ for which the number of messages which can be sent perfectly using entanglement, $c_{\SE}(\mc{N})$, is greater than $c_{0}(\mc{N})$.

\medskip
\emph{Entanglement-assisted zero-error communication.}
	Given a classical channel $\mc{N}$ from Alice and Bob, with inputs $X$ and outputs $Y$, how might they make use of entanglement to increase the number of messages which can be sent? Suppose that Alice wants to send one of $q$ messages to Bob without error and that their entangled shared system is in state $\rho_{AB}$. She will perform some operations on her side of the entangled system, and conditioned on the outcomes of any classical measurements that she does, and on the message $m$ that she wants to send, choose some input to $\mc{N}$. All of this can be represented by saying that she chooses one of $q$ generalised measurements according to $m$, each with $|X|$ outcomes, to perform on her side of the state, and then uses the outcome $k$ as input to $\mc{N}$. Since the residual state on Alice's side is irrelevant to Bob's ability to decode the message, the encoding is fully specified by the POVMs $\{ E^{(m)}_1, \ldots, E^{(m)}_k \}$ for $m \in [q] := \{1, \dots, q\}$ corresponding to the $q$ different generalised measurements.

	If Alice sends message $m$, then with probability $p^{(m)}_k$, Alice inputs $k$ and the residual state of Bob's system is $\rho^{(m)}_k = (\tr_A E^{(m)}_k \ox \openone \rho)/p^{(m)}_k$. Letting $\beta^{(m)}_k := p^{(m)}_k \rho^{(m)}_k$, for all messages $m$: $\sum_k \beta^{(m)}_k = \tr_A \rho_{AB} =: \rho_B$ reflecting the fact that without information from the classical channel, Bob has no idea which message Alice sent (i.e.\ causality). Conversely, any set of positive operators $\beta^{(m)}_k$ which satisfy this condition for some $\rho_B$ can be realised by a suitable choice of $\rho_{AB}$ and generalised measurements. Now, including the state of the channel output (we label the system $C$) as well as his half of the entangled system, Bob's state after receiving the channel output $y\in Y$ is $\sigma_m := \sum_{x\in X,y\in Y} \mc{N}(y|x) \proj{y}_{C} \ox \beta^{(m)}_x$. The encoding works if and only if Bob can distinguish perfectly between all the $\sigma_m$, i.e.\ for all $m, m' \in [q]$:
$0 = \tr \sigma_m \sigma_{m'} 
   =\sum_{x,x'\in X\text{ confusable}} \left( \sum_y \mc{N}(y|x)\mc{N}(y|x') \right) 
                                           \tr \beta^{(m)}_x \beta^{(m')}_{x'}$.
We therefore have:
\begin{theorem}\label{SEA}
For any channel $\mathcal{N}$ with inputs $X$ and outputs $Y$, 
$c_{\SE}(\mathcal{N}) = q(G(\mc{N}))$, where $q(G(\mc{N}))$ is the maximum integer $q$ 
such that there exists a density matrix $\rho_B$ and positive semidefinite operators 
$\beta^{(m)}_x$ for all $m \in [q]$, $x\in X$, on some Hilbert space such that for all $m$,
$\sum_{x\in X} \beta^{(m)}_x = \rho_B$, and
\begin{equation*}
	\forall m\neq m'\ \forall \text{ confusable }x,x' \quad \tr \beta^{(m)}_x \beta^{(m')}_{x'} = 0.
\end{equation*}
In particular, $c_{\SE}(\mathcal{N})$ depends only on $G(\mc{N})$.
		\qed
\end{theorem}

	In light of this fact, it is clear that if a channel has no unassisted zero--error capacity then entanglement cannot change this. Otherwise, entanglement would allow perfect communication over the completely noisy channel, in violation of causality!

	However, there are some channels, for which $c_{\SE} > c_0 > 0$. Examples of such channels can be constructed from proofs of the \emph{Kochen-Specker (KS) theorem} \cite{KS}: We call a family $\{ B_m \}_{m=1}^q$ of complete orthogonal bases $B_m$ of $\mathbb{C}^d$ a \emph{KS basis set} if it is impossible to select one vector from each basis such that no two are orthogonal. That such sets exist is a corollary of the KS theorem \cite{KS}.

\begin{theorem}\label{KSchan}
	For any KS basis set $Z = \{ B_m \}_{m=1}^q$ in $\mathbb{C}^d$ consisting of $q$ orthogonal bases, one can construct a classical channel $\mathcal{N}$ with $c_{0}(\mathcal{N}) < q$ and $c_{\SE}(\mathcal{N}) \geq q$.
\end{theorem}
\begin{proof}
	Let us write $B_m = \{ \psi_{m1}, \ldots, \psi_{md} \}$. We can construct a channel $\mathcal{N}_Z$ with inputs in $[q]\times[d]$ such that a pair of inputs $(m,j)$, $(m',j')$ are confusable if and only if the corresponding vectors $\psi_{mj}$ and $\psi_{m'j'}$ are orthogonal. (In general there are many ways to do this and any one will do. For instance, one can add an output symbol for each orthogonal pair which can be activated by both inputs in that pair but no others.) $G(\mc{N})$ has an edge between inputs if and only if the corresponding vectors are orthogonal. As such, the vertices of $G$ can be partitioned into $q$ cliques of size $d$, corresponding to the $q$ bases of $Z$, so the independence number of $G$ is certainly no larger than $q$. If there was an independent set of size $q$ in $G$ it would have to have exactly one vertex in each of the $q$ cliques, but this would select one vector in each of the $q$ bases such that no two are orthogonal, contradicting the assumption on $Z$. Therefore, $c_{0}(\mathcal{N}) < q$.

	To send $q$ messages using entanglement, Alice and Bob can use a maximally entangled state of rank $d$: to send $m$, Alice measures her side of the state in the bases $B_m$ and obtains the outcome $j$ (at random). She inputs $(m,j)$ to the channel. Bob's output tells him that Alice's input was in some particular mutually confusable subset, but by construction, these inputs correspond to mutually orthogonal residual states of his subsystem, so he can perform a projective measurement to determine precisely which input Alice made to the classical channel, and hence which of the $q$ messages she chose to send, with certainty.
\end{proof}

In Figure \ref{fig:p24} we give an example of a KS basis set derived from a proof of the KS theorem due to Peres \cite{Peres-KS}.

\begin{figure}
\includegraphics[scale=0.35]{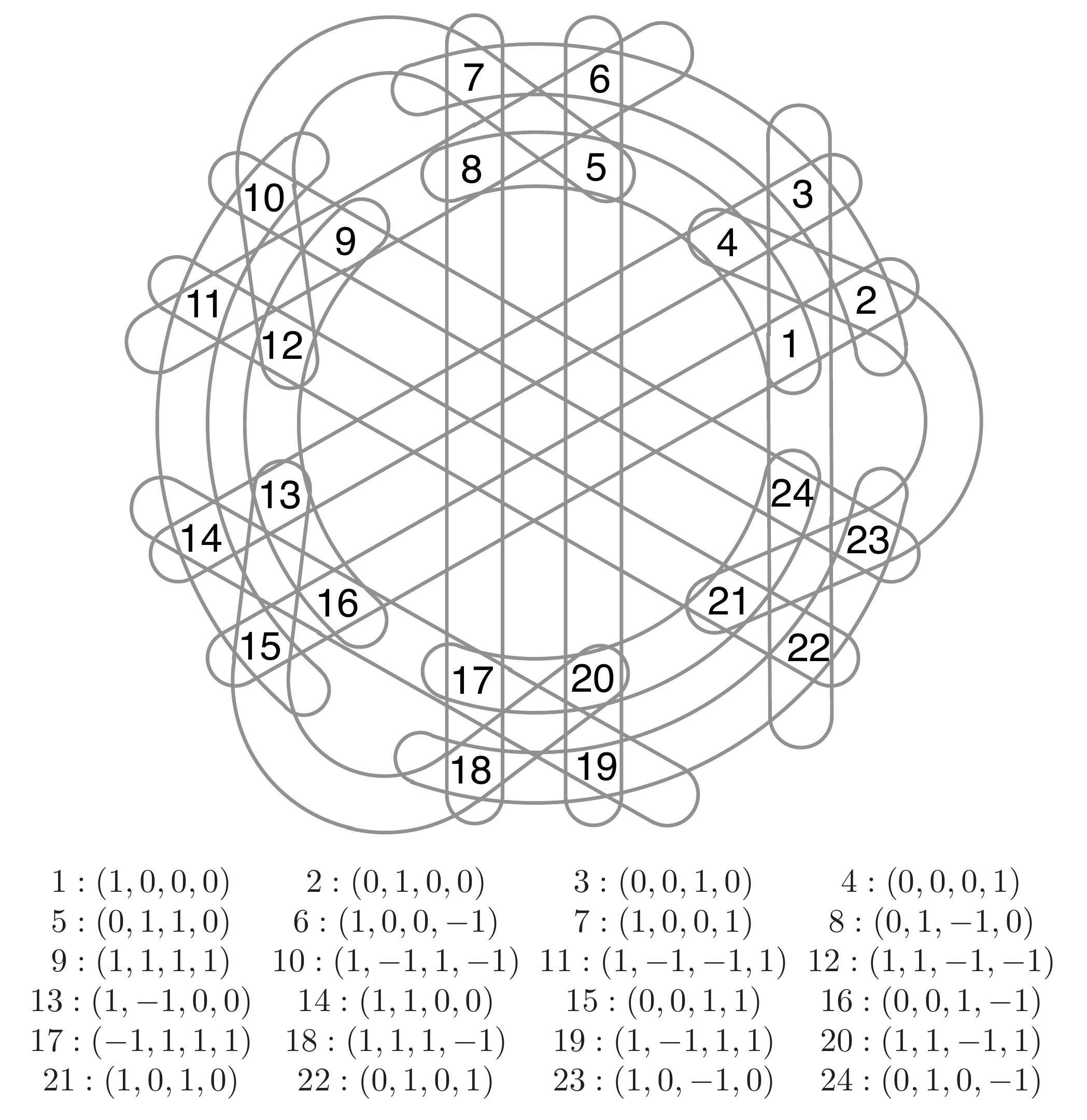}
   \caption{A KS basis set of 6 bases for $\mathbb{C}^4$ is tabulated at the bottom of the figure, one basis per row. The vectors are presented as 4-tuples labelled by a number. The diagram represents a channel $\mc{N}$ with an input symbol for each vector in the set. It has an output symbol for each grey loop: on input $x$ the output is drawn uniformly at random from those corresponding to the 3 loops which contain that $x$. Inputs are confusable if and only if corresponding vectors are orthogonal, so by Theorem 2, $c_0(\mc{N}) < 6$ (in fact it is 5), but $c_{\SE}(\mc{N}) \geq 6$. It is interesting to note that to send one of $6$ symbols (with equal prior probabilities) by a single use of $\mc{N}$, the best \emph{unassisted} code has error probability $\frac{1}{18}$. }
   \label{fig:p24}
\end{figure}

\emph{Relationship to pseudo-telepathy games.} 
This increase of the one-shot zero error capacity is an example of 
performing a classical task without error using entanglement,
that becomes impossible without the entanglement.
This phenomenon might sound familiar to those who have encountered `pseudo-telepathy'
games (hereafter \emph{PT-games}) \cite{PT}. The difference is that in these 
games Alice and Bob are not allowed to communicate with each other at all, 
but instead communicate with a verifier who sends them questions and then decides 
whether or not they win the game based on their replies.

	To be precise, in this context a `game' $\mathfrak{g}$ consists of 
questions $a$ and $b$ (drawn according to a fixed distribution $p(a,b)$) 
to Alice and Bob respectively, who reply with answers $\alpha$ and $\beta$.
These are accepted with probability $A(a,b,\alpha,\beta)$, 
$A$ also being a fixed distribution. 
The probability of acceptance (a.k.a.\ `winning') is given by
\[
	\mathfrak{g}(s) := \sum_{a,b,\alpha,\beta} A(a,b,\alpha,\beta) p(a,b) s(\alpha,\beta|a,b),
\]
where the $\emph{strategy}$ $s(r|q)$ is a correlation describing the responses $r$
of the provers to questions $q$. Note that
$\mathfrak{g}(s)$ is a linear function of $s$. 
We call the strategy $s$ `perfect' (for the game $\mathfrak{g}$) if and only if $\mathfrak{g}(s) = 1$. 
Typically we are interested in the best winning probability which can be 
achieved if the strategy is restricted to some class of correlations 
like $\NS$ or $\SE$. A \emph{PT-game} is a game $\mathfrak{g}$ which can be won 
with certainty by a strategy in $\SE$ but cannot be won with certainty by any 
strategy in $\SR$.

\begin{proposition}
	For any channel $\mc{N}$ with inputs $X$ and outputs $Y$, and integer $n$, 
	there exists a natural game $\mathfrak{g}$ such that $\mathfrak{g}$ has a 
	perfect strategy in the class of correlations $\arbclass$ if and only if 
	$c_{\arbclass} \geq n$.
\end{proposition}
\begin{proof}
	In the game $\mathfrak{g}$, the verifier sends Alice $m \in [n]$ and 
	Bob $y \in Y$ drawn independently and uniformly at random. Alice sends back 
	an answer $x \in X$ and Bob replies with $\widehat{m} \in [n]$. 
	If $\mc{N}(y|x) > 0$ then they win the game if and only if $m = \widehat{m}$. 
	Otherwise, they always win the game. A strategy $s$ is perfect for this game 
	if and only if $\sum_{x,y} \mc{N}(y|x) s(x,\widehat{m}|m,y) = \delta_{m\widehat{m}}$.
	Therefore, there is a perfect strategy for $\mathfrak{g}$ in 
	$\arbclass$ if and only if $c_{\arbclass}(\mc{N}) \geq m$.
\end{proof}

	This means that, in order to give an advantage for zero-error coding over $\SR$, 
a correlation in $\SE$ must also be able to win a particular PT-game with certainty 
(and hence sit on the boundary of the non-signalling polytope).

\medskip
\emph{Non-signalling assisted zero-error capacity and exact simulation.} 
While all correlations which can be realised by measurements on entangled states are non-signalling, the converse is not true, as in the case of the Popescu-Rohrlich box \cite{PR}.
	Consequently, we can study non-signalling assisted protocols to find upper bounds for entanglement assistance, but this study also leads to a beautifully simple theory of non--signalling assisted zero--error communication.

Recalling the definition of a hypergraph,
the \emph{fractional-packing number} $\alpha^*(H)$ of a hypergraph $H$ \cite{fpn} 
on vertices $X$ is the maximum value of $\sum_{x\in X} v(x)$ where 
$v: X \to [0,1]$ weights the vertices subject to the constraint that 
for all hyperedges $S$ of $H$, $\sum_{x\in S}v(x) \leq 1$.

\begin{theorem}
\label{NSZELP}
For a classical channel $\mathcal{N}$ with hypergraph $H(\mathcal{N})$,
  \begin{equation*}
    c_{\NS}(\mathcal{N})
    = \left\lfloor \alpha^*(H(\mathcal{N})) \right\rfloor,
  \end{equation*}
where $\alpha^*(H(\mathcal{N}))$ is the fractional-packing number of $H(\mathcal{N})$. 

Furthermore, since the function $\alpha^*$ is multiplicative, in the sense that $\alpha^*(H(\mc{N}_1\ox\mc{N}_2)) = \alpha^*(H(\mc{N}_1))\alpha^*(H(\mc{N}_2))$, the NS-assisted zero-error capacity of $\mc{N}$ is
  \begin{equation*}
    C_{\NS}(\mathcal{N}) = \log \alpha^*(H(\mathcal{N})),
  \end{equation*}
  which is additive: $C_{\NS}(\mathcal{N}_1\ox\mathcal{N}_2) = C_{\NS}(\mathcal{N}_1) + C_{\NS}(\mathcal{N}_2)$.
\end{theorem}

To get the best upper bounds on entanglement assisted zero-error communication using this result, we should minimise over all hypergraphs with the same confusability graph $G$ as the channel in question, because $c_{\SE}$ depends only on $G$ (see Theorem \ref{SEA}).

The proof of Theorem \ref{NSZELP} is given in \cite{long}. With one interesting proviso, the non-signalling assisted zero-error capacity $C_{\NS}(\mathcal{N})$ is the same as the feedback-assisted zero-error capacity of the channel $C_{\mathrm{0F}}(\mathcal{N})$, as derived by Shannon in his seminal paper \cite{shannon}. The proviso applies only when the unassisted zero-error capacity is zero: 
Then $C_{\NS}$ can be positive, whereas $C_{\mathrm{0F}}$ is always zero. We will now give a simple example of this. Let $\mathcal{N}$ be the classical channel which takes as input $j$ an element of the set $A = \{1,2,3,4\}$, and outputs a $2$-element subset of $A$ which contains $j$. Since any two inputs of this channel can be confused (i.e.\ can lead to the same output), it has no unassisted zero-error capacity.

We now exhibit a bipartite correlation $P(x,y|a,b)$ that can be used to boost the zero error capacity of $\mathcal{N}$ to one bit: Alice's input $a$ is a bit and Bob's input $b$ is a $2$-element subset of $A$. Alice's output $x$ is an element of $A$, drawn uniformly at random (independently of either input); if $x \in b$ then Bob's output $y$ is set to $a$, otherwise it is set to $\rm{NOT}(a)$. Clearly, the marginal distribution of Bob's output is independent of Alice's input and vice versa, so $P$ is non-signalling.

Now, suppose Alice plugs her output of $P$ into the channel $\mathcal{N}$ and Bob uses the output of $\mathcal{N}$ as his input $b$ to $P$. Given the behaviour of $\mathcal{N}$ this forces $b$ to contain $x$, therefore Bob's output $y$ will always be equal to $a$. A bit is transmitted from Alice to Bob with perfect reliability!


\medskip
\emph{Channel simulation and reversibility.}
One can also consider the `reverse' problem to zero-error coding 
\cite{long}, and ask what is the minimum identity channel needed, given correlations in $\arbclass$, to simulate one (or more) uses of some noisy channel $\mc{N}$ exactly (in the sense of exactly reproducing the conditional probability distribution of outputs given inputs).
We denote this minimum required number of messages by $k_{\arbclass}(\mc{N})$, 
and the \emph{$\arbclass$-assisted simulation cost} of $\mc{N}$ by 
$K_{\arbclass}(\mc{N}) := \lim_{n \to \infty} \frac{1}{n} \log k_{\mathrm{NS}}(\mc{N}^{\ox n})$. 
Again, the structure of the set of all 
non-signalling correlations results in a very simple formula for $k_{\rm{NS}}(\mc{N})$:
For any channel $\mc{N}$ with inputs $X$ and outputs $Y$,
$
k_{\rm{NS}}(\mc{N}) = \left\lceil \sum_{y} \max_{x} \mc{N}(y|x) \right\rceil,
$
and since the sum here is multiplicative under tensor products of the channel matrix,
$
K_{\rm{NS}}(\mc{N}) = \log \left( \sum_{y} \max_{x} \mc{N}(y|x) \right).
$

While we have found examples showing an arbitrarily large gap between 
$k_{\rm{NS}}(\mc{N})$ and $k_{\rm{SR}}(\mc{N})$, the gap disappears in 
the limit of many channel uses: 
$K_{\rm{SR}}(\mc{N}) = K_{\rm{SE}}(\mc{N}) = K_{\rm{NS}}(\mc{N})$ \cite{long}.

Curiously, a kind of combinatorial zero-error reversibility exists when non-signalling correlations are freely available: For a given channel hypergraph $H$, the NS-assisted zero-error capacity of channels with hypergraph $H$ is equal to the infimum of the NS-assisted simulation cost for channels with hypergraph $H$ \cite{long}, in analogy to the direct and reverse Shannon theorems \cite{shannon-BIG,QRST}.

\medskip
\emph{Conclusion.} 
We have shown that entanglement can sometimes be used to increase the number of classical messages which can be sent perfectly over classical channels. To upper bound this quantum advantage, we have given a simple formula for the non-signalling assisted capacity as a linear program. These discoveries present many new questions: Firstly, can entanglement improve the \emph{asymptotic} zero-error capacity, compared to no assistance, as we have seen NS correlations can? More generally, can we find a simple expression for the entanglement assisted zero-error capacity in the one shot or asymptotic case? 
Note that while the best general upper bound known on $C_0$ is given by Lov\'{a}sz'
famous $\vartheta$ function~\cite{Lovasz}, it was very recently found (indeed
prompted by our Theorem~\ref{KSchan}) that
$\vartheta$ is still an upper bound on $C_{\rm{SE}}$~\cite{Beigi,DSW}.
Can we find simpler, less contrived, examples of channels where $c_{\SE} > c_0$? 
In another direction, the relationship between BKS theorems and PT-games has 
been studied in \cite{RW}. 
We found connections between the entanglement assisted zero-error phenomenon 
and both of these topics, but left open the development of a fuller understanding 
of the relationships among the three.

\acknowledgments
We would like to thank 
Nicolas 
Brunner, 
Runyao 
Duan, 
Tsuyoshi 
Ito, 
Ashley 
Montanaro,
Marcin 
Paw\l{}owski, 
Paul 
Skrzypczyk and 
Stephanie 
Wehner for discussions.
TSC is supported by a Leverhulme early-career fellowship and the EC 
project ``QAP'' (contract no.~IST-2005-15848).
DL was funded by CRC, CFI, ORF, CIFAR, NSERC, and QuantumWorks.
WM acknowledges the support of NSERC and QuantumWorks.
AW is supported by the EC, the U.K. EPSRC, the 
Royal Society, and a Philip Leverhulme Prize. 
The CQT is funded by the Singapore MoE and the NRF
as part of the Research Centres of Excellence programme.
We are grateful for the hospitality of the Kavli Institute for Theoretical Physics 
at UCSB, where a large part of this research was performed. 
This research was supported in part by the NSF under Grant No.~PHY05-51164.

\end{document}